\documentclass[sigconf]{acmart}

\AtBeginDocument{%
  \providecommand\BibTeX{{%
    Bib\TeX}}}

\copyrightyear{2024}
\acmYear{2024}
\setcopyright{rightsretained}
\acmConference[CIKM '24]{Proceedings of the 33rd ACM International
Conference on Information and Knowledge Management}{October 21--25,
2024}{Boise, ID, USA}
\acmBooktitle{Proceedings of the 33rd ACM International Conference on
Information and Knowledge Management (CIKM '24), October 21--25, 2024,
Boise, ID, USA}
\acmDOI{10.1145/3627673.3679925}
\acmISBN{979-8-4007-0436-9/24/10}

\makeatletter
\gdef\@copyrightpermission{
  \begin{minipage}{0.3\columnwidth}
  \href{https://creativecommons.org/licenses/by/4.0/}{
  \includegraphics[width=.9\columnwidth]{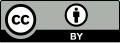}
  }
  \end{minipage}\hfill
  \begin{minipage}{0.7\columnwidth}
   \href{https://creativecommons.org/licenses/by/4.0/}{This work is licensed under a Creative Commons Attribution International 4.0 License.}
  \end{minipage}
  \vspace{5pt}
}
\makeatother

\usepackage{balance}

\usepackage{hyperref}
\hypersetup{
    colorlinks=true,    
    linkcolor=blue,     
    citecolor=blue,     
    urlcolor=blue       
}
\usepackage{graphicx}
\usepackage{multirow}
\usepackage{booktabs}
\usepackage{array}
\usepackage{amsthm}
\usepackage{amsmath}
\usepackage{caption}
\newcolumntype{C}[1]{>{\centering\arraybackslash}p{#1}}

\newtheorem{definition}{Definition}[section]
\newtheorem{theorem}{Theorem}[section]
\newtheorem{corollary}{Corollary}[theorem]

\newtheorem{lemma}[theorem]{Lemma}

\newcommand{\BibTeX}{B\kern-.05em{\sc i\kern-.025em b}\kern-.08em\TeX}

\begin{document}

\title[Facets of Disparate Impact]{Facets of Disparate Impact:\\Evaluating Legally Consistent Bias in Machine Learning}

\author{Jarren Briscoe}
\email{jarren.briscoe@wsu.edu}
\orcid{0000-0002-7422-9575}
\affiliation{%
  \institution{Washington State University}
  \city{Pullman}
  \state{Washington}
  \country{USA}
}

\author{Assefaw Gebremedhin}
\email{assefaw.gebremedhin@wsu.edu}
\orcid{0000-0001-5383-8032}
\affiliation{%
  \institution{Washington State University}
  \city{Pullman}
  \state{Washington}
  \country{USA}
}

\begin{abstract}
Leveraging current legal standards, we define bias through the lens of marginal benefits and objective testing with the novel metric ``Objective Fairness Index". This index combines the contextual nuances of objective testing with metric stability, providing a legally consistent and reliable measure. Utilizing the Objective Fairness Index, we provide fresh insights into sensitive machine learning applications, such as COMPAS (recidivism prediction), highlighting the metric's practical and theoretical significance. The Objective Fairness Index allows one to differentiate between discriminatory tests and systemic disparities.
\end{abstract}



\maketitle

\section{Introduction}

In the 2010s, machine learning took major leaps forward in everyday life. However, we soon realized that bias is rampant in machine learning models. One of the most prominent examples regarding COMPAS, which deals with Americans' civil rights in the United States justice system \cite{brennan2009evaluating,dieterich2016compas}. Researchers are actively working to mitigate bias in general in a myriad of ways \cite{debiasDai,debiasDamak,debiasPrivacy,mehrabi2021survey,caton2020fairness,corbett2018measure,zhang2020machine,fedDebias,selectionBiasKwon,russo2023causal,zhang2023individual,cornacchia2023counterfactual,rec4adBiasGao,abTestBiasLe,nnPruningBriscoe}. Despite being the cornerstone of many proposals, definitions of bias remain inconsistent \cite{caton2020fairness,verma2018fairness}. Due to various philosophies and situations, there is no consensus on a formal definition. We focus on the legal context, namely the disparate impact metric (DI). DI is useful as a flag or reparative measure, however it is insufficient in objective-testing laws.

To address this gap, we introduce the Objective Fairness Index (OFI), extending our work in \cite{fdiOriginal}. We formalize OFI by defining bias as the difference between marginal benefits, derived via the lens of objective testing. In law, objective testing is the philosophy that correct assessments found via rational, related tests are not discriminatory. Additionally, we provide applications to two critical bias scenarios: COMPAS and Folktable's Adult Employment dataset.

The principle of objective testing in non-discrimination laws is recognized not only in the US but worldwide. According to \cite{euObjTesting}, 34 out of 35 surveyed countries (including Spain, Germany, Switzerland, Turkey, and Romania) have such laws in their national legislation. Going beyond the precedents found in the U.S. and the 34 other countries studied in \cite{euObjTesting}, we find that South Africa (Employment Equity Act), Canada (Public Service Employee Relations Commission v British Columbia Government Service Employees' Union), and Australia (Fair Work Act) uphold similar legal standards. Our work with the Objective Fairness Index aims to establish an agreeable and naturally interpretable definition of bias, grounded in these comprehensive legal frameworks and precedents.

We make four key contributions in this work:
\begin{itemize}
    \item We substantially support the Objective Fairness Index with insights from the legal literature.
    \item We show that the disparate impact metric lacks the context of objective testing, which is integral in OFI.
    \item We bring novel insights into two popular ML fairness applications: COMPAS (predicts recidivism) and Folktable's Adult Employment (predicts employment).
    \item We open source our code \href{https://github.com/jarrenbr/objective-fairness-index}{here}\footnote{\href{https://github.com/jarrenbr/objective-fairness-index}{https://github.com/jarrenbr/objective-fairness-index}}. The code gives an efficient solution to handling many confusion matrices, creates the empirical results, and analyzes OFI's behavior.
\end{itemize}
\section{Limitations of Disparate Impact}\label{sect:facetsOfDi}
This section gives context to the founding of the Objective Fairness Index by reviewing the Supreme Court of the United States' decision that spurred disparate impact and a following precedent. Under objective-testing laws, we show that OFI is better than DI in identifying discriminatory tests.

\subsection{Legal Foundations of Disparate Impact}
Disparate impact's significance comes from United States labor laws regarding the biased treatment of protected groups. In Griggs v. Duke Power Co., 401 U.S. 424 (1971), the Supreme Court unanimously agreed that employers must prove that any test or assessment used must be related to job performance (objective). Written by Chief Justice Burger, the decision says that a lower admittance rate of a protected class warrants further investigation---even if such discrimination may be unintentional. Since this precedent was passed from a lawsuit under Title VII of the Civil Rights Act of 1964, the Court's ruling extends to decision-making whenever protected classes are concerned.

To summarize this ruling, employers must have objective and fair tests (neutral in form) and should not inadvertently advantage a protected class (neutral in impact). Should there be a disparate impact, the employer is subject to scrutiny and bears the burden of proof that the test is objective and ``consistent with business necessity''. The disparate impact bias metric is defined in Definition \ref{def:di}. Where
$P=FN+TP$ denotes the amount of positive labels, $\hat{P}=FP+TP$ denotes the amount of positive predictions, $N=FP+TN$ denotes the amount of negative labels, and $\hat{N}=FN+TN$ denotes the amount of negative predictions.
\begin{definition}\label{def:di}
Disparate impact compares the rates of positive outcomes between two groups using ratios.
\begin{align}\label{eq:di}
    DI \triangleq \hat{P}_i/n_i \div \hat{P}_j/n_j
\end{align}
\end{definition}
The implication of Equation \eqref{eq:di} is that $DI=1$ suggests no bias. However, the real world can have small and non-representative sampling, causing unbiased decisions to be $DI\neq 1$. As such, a Four-Fifths Rule (80\% Rule) has been established in law as an initial screening tool to identify potential discrimination (see 29 CFR § 1607.4 - Information on impact) \cite{biddle2017adverse}. This rule suggests that if $DI > 5/4$ then group $i$ has positive bias, if $DI < 4/5$ then there is positive bias for $j$, otherwise DI finds no bias.

In Griggs v. Duke Power Co., the Court goes on to note that Congress does not intend to require businesses to give jobs to unqualified candidates, and that context should still be considered, emphasizing that truly objective assessments are not discriminatory. However, because disparate impact lacks situational context, it cannot capture this legal nuance.

Disparate impact has been influential enough for employers to discard assessments that fail the heuristic to prevent any federal investigation or social uprising, regardless of the assessments' objectivity. Ricci v. DeStefano, 557 U.S. 557 (2009) is the first of these cases to be argued in the Supreme Court when New Haven, Connecticut dismissed objective assessments that benefited the white and Hispanic firefighters over other protected classes. The Court ruled in the firefighters' favor---that omitting an objective performance test solely due to differing acceptance rates (disparate impact) is a violation of the rights of those who correctly benefit from such a test. This provides us with a use case where equal positive prediction rates (which became zero after the test's omission) are discriminatory.

\subsection{Defining Benefit and Bias}
We define bias as a function of what happened to the group (benefit) and what should have happened to the group (expected benefit). An objective test will show that benefit equals expected benefit. Consistent with common understanding and legal precedents, we define benefit as an advantage that is gained from a decision or action, such as hiring from a machine-learning prediction. We achieve formal definitions of these terms using binary confusion matrices, allowing generalization to situations where information such as specific features cannot be used due to privacy, legal protection, or information loss.

We define an individual benefit in Definition \ref{def:indBenefit} to help describe a group's benefit in Definition \ref{def:groupBenefit}. We assume the positive prediction is the beneficial prediction in all definitions.

\begin{figure}[]
    \centering
    \includegraphics[width=.5\linewidth]{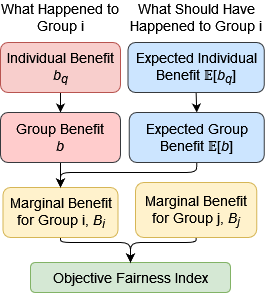}
    \caption{Objective Fairness Index's conceptual framework.}
    \label{fig:ofiDiagram}
      \Description{A flowchart illustrating the Objective Fairness Index's framework.}
\end{figure}

\begin{definition}\label{def:indBenefit}
An individual $q$'s benefit $b_q$ is the individual's prediction. $b_q=1$ if predicted positive, otherwise 0.
\end{definition}
\begin{definition}\label{def:groupBenefit}
A group's benefit is the arithmetic mean of all individual benefits within the group.
\begin{align}\label{eq:groupBenefit}
    b = \frac{1}{n}\sum_{\forall b_q} b_q = \frac{1\cdot \hat{P} + 0\cdot \hat{N}}{n} = \frac{TP+FP}{n}
\end{align}
\end{definition}
\begin{table*}[h]
    \centering
    \caption{Comparison of OFI, DI, and Law (Ricci v. DeStefano) Across Different Scenarios}
    \begin{tabular}{c c c c c c c c c c c c c c}
         \toprule
         Group & Scenario & TP & FN & FP & TN & $n$ & $b$ & $\mathbb{E}[b]$ & $\mathcal{B}$& \multicolumn{1}{c}{OFI: $\mathcal{B}_i-\mathcal{B}_j$} & \multicolumn{1}{c}{DI: $b_i/b_j$} & Ricci v. DeStefano \\
         \midrule
         $i$ & A& 1 & 0 & 0 & 5 & 6 & $1/6$ & $1/6$ & 0 & \multirow{2}{*}{$-0.06$} & \multirow{2}{*}{$0.38$} & \multirow{2}{*}{Likely no bias}\\
         $j$ & A & 7 & 0 & 1 & 10 & 18 & $8/18$ & $7/18$ & $1/18$&&\\
         \hline
         $i$ & B& 0 & 1 & 0 & 5 & 6 & $0/6$ & $1/6$ & $-1/6$&\multirow{2}{*}{$0.22$} & \multirow{2}{*}{NaN (1 by context)}&\multirow{2}{*}{Bias for $i$}\\
         $j$ & B & 0 & 7 & 0 & 11 & 18 & $0/18$ & $7/18$ & $-7/18$&&\\
         \hline
         $i$ & $\alpha$& 1 & 1 & 0 & 5 & 7 & $1/7$ & $2/7$ & $-1/7$&\multirow{2}{*}{$0.23$} & \multirow{2}{*}{2.71} & \multirow{2}{*}{Bias for $i$}\\
         $j$ & $\alpha$ & 1 & 7 & 0 & 11 & 19 & $1/19$ & $8/19$ & $-7/19$&&\\
         \bottomrule
    \end{tabular}
    \label{tab:margBenHypotheticals}
\end{table*}
The Objective Fairness Index also integrates the expected benefit of an individual (Definition \ref{def:indExpBenefit}) and the group's expected benefit (Definition \ref{def:expGroupBenefit}) by using the actual labels.

\begin{definition}\label{def:indExpBenefit}
An individual's expected benefit $\mathbb{E}[b_q]$ is their actual label. $\mathbb{E}[b_q]=1$ if labeled positive, otherwise 0.
\end{definition}

\begin{definition}\label{def:expGroupBenefit}
A group's expected benefit is the arithmetic mean of all individual expected benefits within the group.
\begin{align}\label{eq:expBenefit}
   \mathbb{E}[b] = \frac{1}{n}\sum_{\forall b_q} \mathbb{E}[b_q] = \frac{1\cdot P- 0\cdot N}{n} = \frac{TP + FN}{n}
\end{align}
\end{definition}

Formally defining benefit and expected benefit allows a counterfactual test of ``what happened'' versus ``what should have happened''. As mentioned, this is of paramount importance in law. In Ricci v. DeStefano, the city of New Haven omitted their test solely based on the ratio of benefits to be obtained by each protected class of firefighters. However, this discriminates against the firefighters qualified by this objective test. Pursuant to the Supreme Court's ruling, we present a solution by scoring this group's treatment as the ``marginal benefit'' ($\mathcal{B}$ in Definition \ref{def:margBenefit}).

\begin{definition}\label{def:margBenefit}
Marginal benefit ($\mathcal{B}$) is the difference between a group's benefit and expected benefit.
\begin{align}\label{eq:margBenefit}
    \mathcal{B} = b - \mathbb{E}[b] = (FP-FN)/n
\end{align}
\end{definition}

This marginal benefit score suggests: $\mathcal{B}<0$ an unjust lack of benefit; $\mathcal{B}=0$ an appropriate amount of benefit; $\mathcal{B}>0$ an unjust surplus of benefit. With marginal benefit defined, we now formalize the definition of the Objective Fairness Index in Definition \ref{def:ofi}. We provide a graphical representation in Figure \ref{fig:ofiDiagram}.
\begin{definition}\label{def:ofi}
    The Objective Fairness Index is the difference between two groups' marginal benefits.
    \begin{align}\label{eq:ofi}
        OFI = \mathcal{B}_i - \mathcal{B}_j = \frac{FP_i-FN_i}{n_i} - \frac{FP_j - FN_j}{n_j} \in [-2,2]
    \end{align}
\end{definition}

\subsection{Improvement over Disparate Impact}\label{sect:ofiImproveDi}
To illustrate the Objective Fairness Index's significance, let us consider the instances in Table \ref{tab:margBenHypotheticals}, similar to the case of New Haven and its firefighters. The same objective test is given in all scenarios. In Scenario A, the employer fulfills their promise and promotes (OFI=-0.06, DI=0.38). In Scenario B, the test is discarded (OFI=0.22, DI=1). Note Scenario B's DI score is undefined. However, since the philosophy of DI suggests equal positive prediction rates indicate no bias, we use DI $=1$.

With OFI's absolute value being greater in Scenario B than in Scenario A, OFI is consistent with the ruling of Ricci v. DeStefano (2009). However, DI disagrees with this precedent, suggesting for strong bias in Scenario B and none in Scenario A.

In Scenario $\alpha$, we conduct a thought experiment with a smoothing technique. Scenario $\alpha$ adds one qualified individual to each group that benefits (e.g., is correctly promoted) to Scenario B. Adding this single correct prediction changes DI's score from 1 to 2.71, suggesting a transition from no bias to heavy bias for group $i$. However, OFI adjusts from 0.22 to 0.23 in the same Scenario $\alpha$, illustrating its robustness.

Unlike DI, the Objective Fairness Index indicates the correct directional bias in both scenarios. Additionally, OFI's defined-everywhere property and robustness against small changes showcase its stability. However, we note that DI is complementary to OFI in systemic disparity detection. With OFI, we can identify if an algorithm is at fault (OFI $\not\approx 0$) or if a systematic disparity exists outside of the algorithm (OFI $\approx 0$ and DI $\not\approx 1$).

\subsection{Establishing OFI Thresholds}
Often, data has noise and perfectly objective tests are an open question. As such, we use thresholds. Some thresholds are given by a rule of thumb, such as DI's Four-Fifths Rule. However, we formally derive OFI's threshold to an upper bound of $OFI \in [-0.3,0.3]$. We obtain these bounds by calculating the standard deviation of OFI over all confusion matrices (see the proofs in our code repository). We stress that the threshold ought to be reduced when we expect better than uniform confusion-matrix distributions, such as in predicting recidivism.
\section{Applying OFI: Empirical Case Studies}
We explore the application of the Objective Fairness Index in practical scenarios: COMPAS (predicts recidivism), and Folktables' Adult Employment dataset (predicts employment).

For COMPAS, we obtain the confusion matrix from actual COMPAS predictions and cases of recidivism occurring within two years \cite{propublica_compas,compasCms}. Since we assume the positive label is beneficial, we flip the COMPAS positive label to be ``not a recidivist".

Figure \ref{fig:compas} shows bias scores for each ethnicity pair in the COMPAS dataset. Here, we see that OFI affirms DI's findings. This is an important revelation, as some argue that the disparate prediction rates are correct. In return, some argue that the labeled dataset suffers from systematic bias. However, this case illustrates that OFI can sidestep the need to prove incorrect ground labels, and shows that even if one assumes correct ground labels, there is algorithmic bias against certain races in COMPAS (OFI $\not\approx 0$). We believe this revelation will encourage more people to take corrective measures.

For Folktable's dataset, we analyze data from the Georgia census (393,236 observations) to predict employment status using a Random Forest classifier (RF) and Na\"ive Bayes (NB). In RF, OFI corroborates the findings from DI. However, OFI reveals differences to DI in the NB experiment (see Figure \ref{fig:folk_nb}). Here, DI suggests a positive bias for Pacific Islanders in 7/8 cases. In contrast, OFI shows positive bias in 2/8 cases.

With the lens of objective testing, we find that Pacific Islanders suffer from algorithmic bias, despite having high disparate impact scores. These studies demonstrate OFI confirming suspicions of algorithmic bias (COMPAS), and when protected classes suffer from algorithmic bias, despite a high positive prediction rate.

\begin{figure}
    \centering
    \includegraphics[width=.76\linewidth]{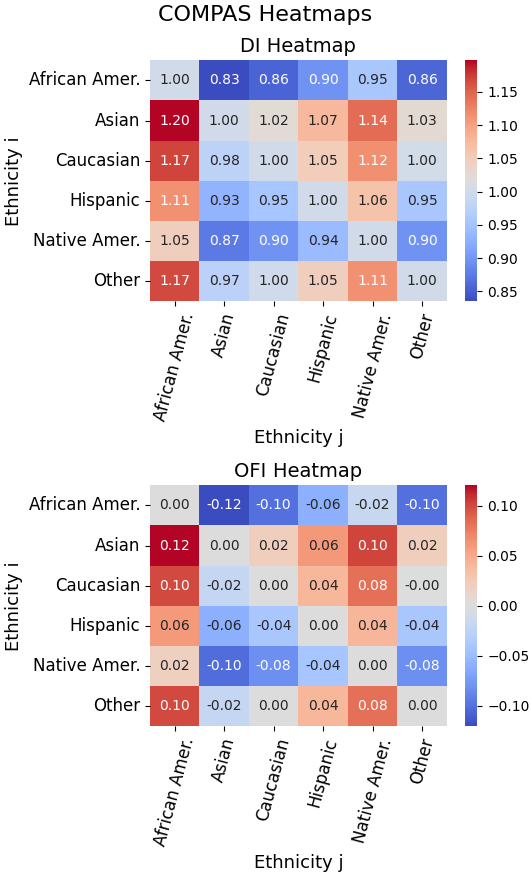}
    \caption{OFI shows that COMPAS manifests algorithmic bias in addition to manifesting disparate impact (DI). We use the predictions and labels published in \cite{propublica_compas}.}
    \label{fig:compas}
    \Description{Two heatmaps showing disparate impact and objective testing index scores for all ethnicity pairs in the COMPAS dataset.}
\end{figure}
\begin{figure}
    \centering
    \includegraphics[width=.87\linewidth]{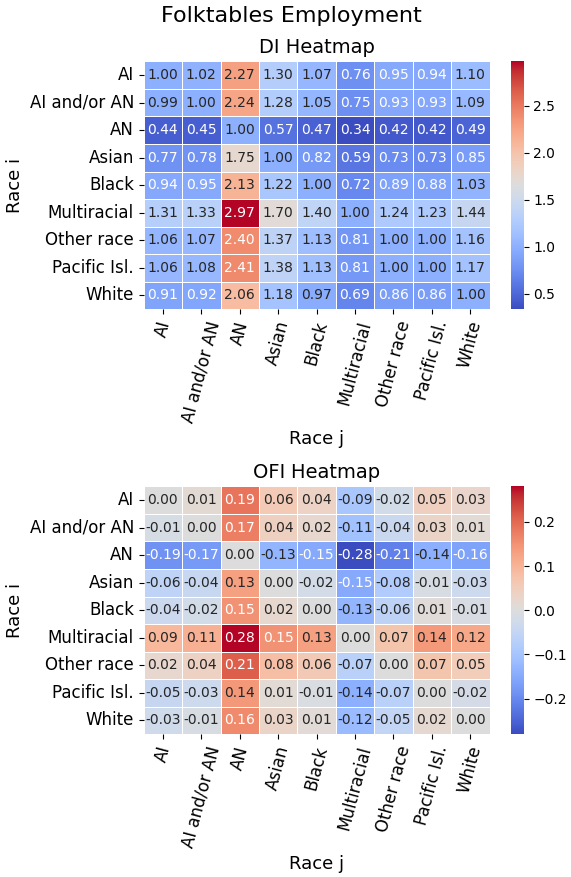}
    \caption{DI conveys that Pacific Islanders have positive bias for 7/8 comparisons. However, OFI shows that Pacific Islanders suffer from algorithmic bias in 6/8 cases. Key: AI is American Indian, AN is Alaska Native, Pacific Isl. is Pacific Islander.}
    \label{fig:folk_nb}
        \Description{Two heatmaps showing disparate impact and objective testing index scores for all race pairs in the Folktables' Employment dataset for the 2014-2017 Georgia census.}
\end{figure}

\section{Conclusions and Future Directions}
In this work, we introduce the Objective
Fairness Index (OFI) to evaluate bias in a legally grounded and context-aware perspective. By leveraging legal frameworks and precedents, particularly those stemming from objective testing and disparate impact theory, OFI captures the nuanced understanding of bias that encompasses both what occurred and what ought to have occurred.

In Section \ref{sect:facetsOfDi}, we illustrate that the Objective Fairness Index is legally grounded, even where the disparate impact metric fails. Specifically, OFI finds no bias in objective testing while DI may or may not, and OFI can correctly detect bias where DI does not. Our analyses include legal-case-driven situations and popular bias problems, including COMPAS (risk of recidivism).

The Objective Fairness Index addresses a significant gap in the machine learning and AI fairness literature by offering a metric that not only aligns with legal standards but also provides novel insights into binary classification problems. Our findings highlight the limitations of existing metrics, such as disparate impact, which lack context and may lead to misinterpretation or oversight of critical bias factors. OFI offers a solution by integrating objective testing directly into the evaluation of bias.

For future work, we envision OFI in more complex scenarios, including multi-class classification problems, regression tasks, and recommender systems. Another promising avenue is the exploration of OFI's role in the development of debiasing techniques. 

Additionally, corrective procedures should be investigated. As an objective measure, OFI identifies discriminatory selections, whereas reparative measures such as DI are for identifying social disparities. Thus, if the OFI fails, correction at the algorithmic or employer level is appropriate. However, if OFI shows no bias but reparative measures do, this suggests systematic disparity and can rationalize social programs. We recommend that all measures and corrective procedures are routinely scrutinized regarding ethics and law.

There are some inherent limitations to applying concepts in law to binary classification. Examples include the ``good faith'' of the test-maker (Canada's Meiorin test), proportional responses (UK's Equality Act 2010, suggesting a non-binary decision), and expectations of ``reasonable accommodations'' for certain situations (Americans with Disabilities Act). We hope to explore solutions for these in more complex scenarios where they are addressable.

Our work underscores the importance of interdisciplinary approaches in tackling bias in AI. By grounding our metric in well-established legal precedents and demonstrating its applicability to contemporary machine learning challenges, we contribute a tool for researchers striving for fairness and transparency in AI systems. 

\begin{acks}
This work was supported by a Department of Education GAANN Fellowship and by cooperative agreement CDC-RFA-FT-23-0069 from the CDC's Center for Forecasting and Outbreak Analytics.
Its contents are solely the responsibility of the authors and do not necessarily represent the official views of the Centers for Disease Control and Prevention or DEd.
\end{acks}

\bibliographystyle{ACM-Reference-Format}
\balance
\bibliography{main}

\clearpage
\appendix

\section{Preliminaries}
This section goes over the preliminaries for our counting and distribution analyses. We begin by redefining the confusion matrix as a row vector for convenience. We abstract away the titles of TP, FN, FP, and TN since all have equal occurrences in the set of all possible confusion matrices. Please see \cite{briscoe2025algorithmic} for a continuation of this work, where we thoroughly investigate distributions of classification metrics, discover issues of small-data analyses, and propose solutions.

\begin{definition}
    The confusion matrix, CM, is a quadruple of non-negative integers.
    \begin{align}
        CM \in \mathbb{N}_0^4
    \end{align}
\end{definition}

\begin{definition}\label{def:allCms}
$\mathcal{M}(n)$ is the set of all possible confusion matrices (CMs) for a given amount of samples, $n > 0$.
\begin{align}
    \label{eq:allCms}
    \mathcal{M}(n) = \{ \text{CM} : \forall \text{CM} \in \mathbb{N}_{0}^4 
    \text{ s.t. } \sum_{c\in \text{CM}}c = n.\}
\end{align}
\end{definition}
We give an example with the set $\mathcal{M}(1)$ in \eqref{eq:allCmsNEquals1}.
\begin{align}\label{eq:allCmsNEquals1}
    \mathcal{M}(1) = \{[1,0,0,0], [0,1,0,0], [0,0,1,0], [0,0,0,1]\}
\end{align}
We move from $\mathcal{M}(1)$ to $\mathcal{M}(2)$ and begin to see counting patterns. We denote the count of the value $x$ for any given cell and $n$ samples as $C(x;n)$. For example, $\mathcal{M}(1)$ shows that $C(1;1)=1$ and $C(0;1)=3$. As we continue this walkthrough, we identify triangular numbers and generalize using said numbers. Note that 1 and 3 are the first two triangular numbers. We follow Knuth's notation for the $w^{th}$ triangular number in \eqref{eq:wTermial} \cite{knuthTermial}.
\begin{align}\label{eq:wTermial}
    w?=\sum_{j=1}^{w}j
\end{align}
Next, we show that the triangular sequence continues for $n=2$ by enumerating in \eqref{eq:m2} and counting in \eqref{eq:m2Counts}. Then, we generalize the relation of $x$ to $n$ in Theorem \ref{thm:ordinalCount}.
\begin{align}\label{eq:m2}
    \mathcal{M}(2) = \begin{cases}
         [2,0,0,0],[0,2,0,0],[0,0,2,0],[0,0,0,2],  \\
         [1,0,0,1],[1,0,1,0],[1,1,0,0], [0,1,0,1],\\
         [0,1,1,0],[0,0,1,1].
    \end{cases}\\\label{eq:m2Counts}
     \big(C(2;2),C(1;2),C(0;2)\big)=\big(1?,2?,3?\big)=\big(1,3,6\big)
\end{align}
\section{Counting}
\begin{theorem}\label{thm:ordinalCount}
    Given the number of samples, $n$, the count of value $x\in[0,n]$, denoted as $C(x;n)$, appearing in $\mathcal{M}(n)$ is found by
    \begin{align}\label{eq:countXGivenN}
        C(x;n) = 0.5(n-x+1)(n-x+2) = (n-x+1)?
    \end{align}
\end{theorem}
\begin{proof}
    First, we prove that $C$ follows the triangular sequence. Since all cells are interdependent and required to sum to $n$, we use the stars and bars method for $n^*$ stars and $k-1$ bars \eqref{eq:starsBars}.
    \begin{align}\label{eq:starsBars}
        \binom{n^*+k-1}{k -1}
    \end{align}
    If cell $c$ has value $x$, then the other three cells must comprise of three non-negative integers summing to $n-x$. In the stars and bars terminology, there are $n-x$ stars and $k=3$ other buckets/cells to choose from.
    \begin{align}
        C(x;n) = \binom{n-x+3-1}{3-1} = \binom{n-x+2}{2}
    \end{align}
    Rewriting $C(x;n)$ from the binomial coefficient, we create \eqref{eq:countIntermediate}.
    \begin{align}\label{eq:countIntermediate}
        C(x;n)=\frac{(n-x+2)!}{2!(n-x)!} = 0.5(n-x+1)(n-x+2)
    \end{align}
    Now we prove that $C(x;n)$ is a triangular number:
    \begin{align}\label{eq:countTriangle}
        0.5(n-x+1)(n-x+2) = (n-x+1)?
    \end{align}
    It is well established that the $w^{th}$ triangular number is found by
    \begin{align}
        w? = 0.5w(w+1).
    \end{align}
    So by substituting $w=n-x+1$, we see that \eqref{eq:countTriangle} follows the formula for $w?$.
\end{proof}

\begin{corollary}
    The number of possible combinations for any unique CM for a given $n$ is the cardinality of $\mathcal{M}(n)$, denoted as $\mathcal{N}(n)$. $\mathcal{N}(n)$ can be found in a similar fashion to $C(x;n)$ with $n^*=n$ and $k=4$. $k=4$ since there are four cells to choose from.
    \begin{align}\label{eq:totalCombs}
        \mathcal{N}(n) = \vert \mathcal{M}(n)\vert = \binom{n + 4-1}{4-1} = \frac{(n+1)(n+2)(n+3)}{6}
    \end{align}
\end{corollary}
\begin{corollary}\label{cor:countAllEqualSum}
    The sum of the counts of all values for a given cell is the total number of possible combinations.
    \begin{align}\label{eq:countAllEqualSum}
       \mathcal{N}(n) = \sum_{\alpha=0}^{n}C(\alpha;n)
    \end{align}
\end{corollary}

\section{A Quasi-Triangular Distribution}
We move to prove that $\mathcal{B}$ mimics a triangular distribution, but will not ever be triangular.

\begin{figure}
    \centering
    \includegraphics[width=\linewidth]{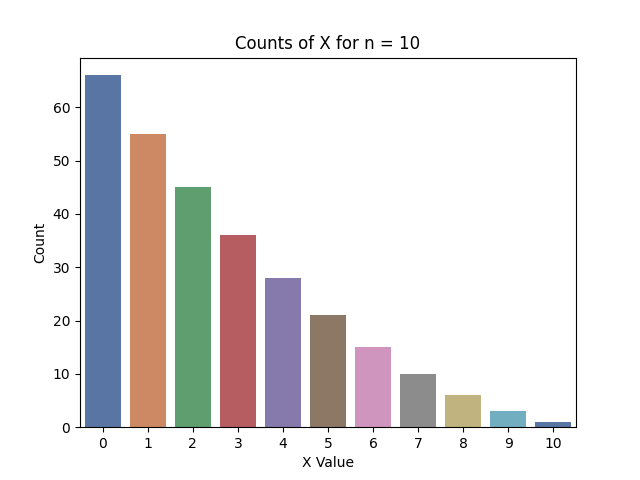}
    \caption{The number of times the value $x$ appears in $\mathcal{M}$.}
    \label{fig:countXGivenN}
\end{figure}

\begin{lemma}\label{thm:countxInc}
    The count of value $x$ increases by $n-x+2$ for each increment of $n$. We illustrate this in Figure \ref{fig:countXGivenN}.
\end{lemma}
\begin{proof}
    From Theorem \ref{thm:ordinalCount}, we prove that $C(x;n) = (n-x+1)?$.
    As such,
    \begin{align}
        C(x;n+1) = (n-x+2)?\\
        C(x;n+1) - C(x;n) = n-x+2.
    \end{align}
\end{proof}

\begin{theorem}
    $\mathcal{B}$'s distribution is not triangular, although it mimics a triangular distribution as $n$ approaches $\infty$.
\end{theorem}
\begin{proof}
    A triangular distribution is defined by a peak (mode), a minimum $a$, a maximum $b$, and a linear increase from $a$ to the mode and from $b$ to the mode. $\mathcal{B}$'s minimum and maximum are trivially -1 and 1. We will now prove that $\mathcal{B}$'s mode is zero.

    The distribution of $\mathcal{B}$ is calculated by collecting all scores, including duplicates. We name this distribution $d_\mathcal{B}$.
    \begin{align}
        d_\mathcal{B} = \{\mathcal{B}(FP,FN), \forall CM \in \mathcal{M}
    \end{align}

    We define the distribution to highlight that every valid confusion matrix is considered. Furthermore, $\mathcal{B}$ is the difference between two values with the same distribution, subject to the constraint $FP+FN\leq n$. As such, we find that equivalent values are the most common pairs. The difference between equivalent values is zero, hence the mode is zero.

    If $d_\mathcal{B}$ was triangular, the final step would be to prove its piece-wise linear property with the endpoints of extrema and mode. For $d_\mathcal{B}$, this is to prove that each total (sum of all counts yielding a score) increments linearly with respect to the score. We found this difficult since some values ($x$) for a finite $n$ have their multiplicity decrease by one from the prior value's ($x-1$) multiplicity when the trend should be overall increasing. This discrepancy becomes exceedingly minuscule as we scale $n$ as evident.

    But, we prove that $\mathcal{B}$ is not a triangular distribution by contradiction. Should $\mathcal{B}$ be triangular, then its standard deviation would be found by $1/\sqrt{6}$ since it is symmetric with bounds $[-1,1]$ \cite{triDist}. However, we show in Theorem \ref{thm:stdB} that $d_\mathcal{B}$'s standard deviation converges to $1/\sqrt{10}$.
\end{proof}

\section{Standard Deviation and Mean of $\mathcal{B}$}\label{sect:bProofs}
We find the mean $\bar{\mathcal{B}}$, variance $\sigma^2$, and standard deviation $\sigma$ for the marginal benefit $\mathcal{B}$. Additionally, we provide their limits/convergences. We begin with $\mathcal{B}$'s mean in Lemma \ref{lem:ofiuMean} as it is a necessary preliminary for the variance and standard deviation.

\begin{lemma}\label{lem:ofiuMean}
    $\mathcal{B}$ has a mean of zero.
    \begin{align}
    \bar{\mathcal{B}} = 0
    \end{align}
\end{lemma}
\begin{proof}
    $\mathcal{B}$ is equivalent to $\mathcal{B}(FP,FN;n)$, a function of two cells in CM, parameterized by the count of all cell values, $n$. With the input space $\mathcal{M}(n)$, for any possible elements of FP, there exists the same element of FN at the same multiplicity. For simplicity, let 
    \begin{align}
        FP = \{a_1, a_2, \dots, a_i, \dots, a_{\mathcal{N}(n)}\}\\
        FN = \{a_1, a_2, \dots, a_i, \dots, a_{\mathcal{N}(n)}\}
    \end{align}
    Then we use the Associative Law of Summation.
    \begin{align}
        \bar{\mathcal{B}}\vert n = \frac{1}{\mathcal{N}(n)}\sum_{i=1}^{\mathcal{N}(n)}\frac{a_i-a_i}{n} = 0
    \end{align}
\end{proof}

We now move to Theorem \ref{thm:stdB} to show how to calculate the standard deviation of $\mathcal{B}$ given any $n$ in constant time. After, we give the limit of $\sigma(\mathcal{B};n)$ in Corollary \ref{cor:bStdCon}. Next, we give the variance and its limit in Corollaries \ref{cor:bvar} and \ref{cor:bvarCon}.

\begin{theorem}\label{thm:stdB}
    The standard deviation of $\mathcal{B}\vert n$ is given by
    \begin{align}
        \sigma(\mathcal{B};n) = \sqrt{\frac{n+4}{10n}}
    \end{align}
\end{theorem}
\begin{proof}
We first cite the formula for calculating the standard deviation for a set $S$ with cardinality $\mathcal{N}$ and mean $\bar{s}$.
\begin{align}
    \sigma = \sqrt{\frac{\sum_{s_i \in S}(s-\bar{s})^2}{\mathcal{N}}}
\end{align}
In $\mathcal{B}$, each score $s\in S$ is calculated by iterating over all confusion matrices (the set $\mathcal{M}$) and mapping through function $\mathcal{B}$. To find all the confusion matrices in $\mathcal{M}$, we take three loops (summations). The fourth cell (TN) is equal to the remainder of $n$ minus the three cells so it is not looped over.

Furthermore, from Lemma \ref{lem:ofiuMean}, the mean for $\mathcal{B}$ is $\bar{\mathcal{B}}=0$. So, $\mathcal{B}$'s standard deviation is given by
    \begin{align}
        \sigma(\mathcal{B};n)
        = \sqrt{\frac{1}{\mathcal{N}(n)}\sum_{TP=0}^n\left(\sum_{FP=0}^{n-TP}\left(\sum_{FN=0}^{n-TP-FP}\left(\mathcal{B}(FP, FN;n) - \bar{\mathcal{B}}\right)^2\right)\right)}\notag\\
        = \sqrt{\frac{1}{\mathcal{N}(n)}\sum_{TP=0}^n\left(\sum_{FP=0}^{n-TP}\left(\sum_{FN=0}^{n-TP-FP}\left(\frac{FP-FN}{n}\right)^2\right)\right)}
    \end{align}
    The summations inside the radical simplify to
    \begin{align}
        \sum_{TP=0}^n\left(\sum_{FP=0}^{n-TP}\left(\sum_{FN=0}^{n-TP-FP}\left(\frac{FP-FN}{n}\right)^2\right)\right) \notag\\
        = \frac{(n+1)(n+2)(n+3)(n+4)}{60n}
    \end{align}
Using the definition of $\mathcal{N}$ from \eqref{eq:totalCombs}, 
\begin{align}
    \sigma(\mathcal{B};n) =\notag\\
    \sqrt{\frac{1}{(n+1)(n+2)(n+3)/6} \cdot \frac{(n+1)(n+2)(n+3)(n+4)}{60n}}\notag\\
    = \sqrt{\frac{n+4}{10n}}
\end{align}
\end{proof}

\begin{corollary}\label{cor:bStdCon}
The standard deviation of $\mathcal{B}$ converges to $1/\sqrt{10}\approx 0.316$.
    \begin{align}
    \lim_{n\rightarrow\infty}\sigma(\mathcal{B};n) = \lim_{n\rightarrow\infty} \sqrt{\frac{n+4}{10n}} = \frac{1}{\sqrt{10}}
    \end{align}
\end{corollary}
\begin{corollary}\label{cor:bvar}
    The variance of $\mathcal{B}$ is $(n+4)/10$.
    \begin{align}
        \sigma^2(\mathcal{B};n) = \frac{n+4}{10}
    \end{align}
\end{corollary}
\begin{corollary}\label{cor:bvarCon}
    The variance of $\mathcal{B}$ converges to $1/10$.
    \begin{align}
        \lim_{n\rightarrow \infty}\sigma^2(\mathcal{B};n) = \frac{1}{10}
    \end{align}
\end{corollary}

We use SymPy's \cite{sympy} computer algebra system (CAS) to verify Theorem \ref{thm:stdB} and Corollaries \ref{cor:bStdCon}, \ref{cor:bvar}, and \ref{cor:bvarCon}.

\section{Empirical Studies}
We present four case studies from two well-regarded baselines: COMPAS (risk of recidivism) and Folktable’s Adult Employment (is person employed) datasets. 

In Figure \ref{fig:compas}, we see that OFI affirms preliminary findings of DI. This is an important revelation, as DI is susceptible to arguments and cases of higher prediction rates being correct. In return, some argue that higher rates would be due to systematic bias. However, OFI sidesteps the need to prove incorrect ground labels and shows that even if one assumes correct ground labels, there’s algorithmic bias against certain ethnicities. We believe this revelation will encourage more people to take corrective measures.

\begin{figure*}
    \centering
    \includegraphics[width=\linewidth]{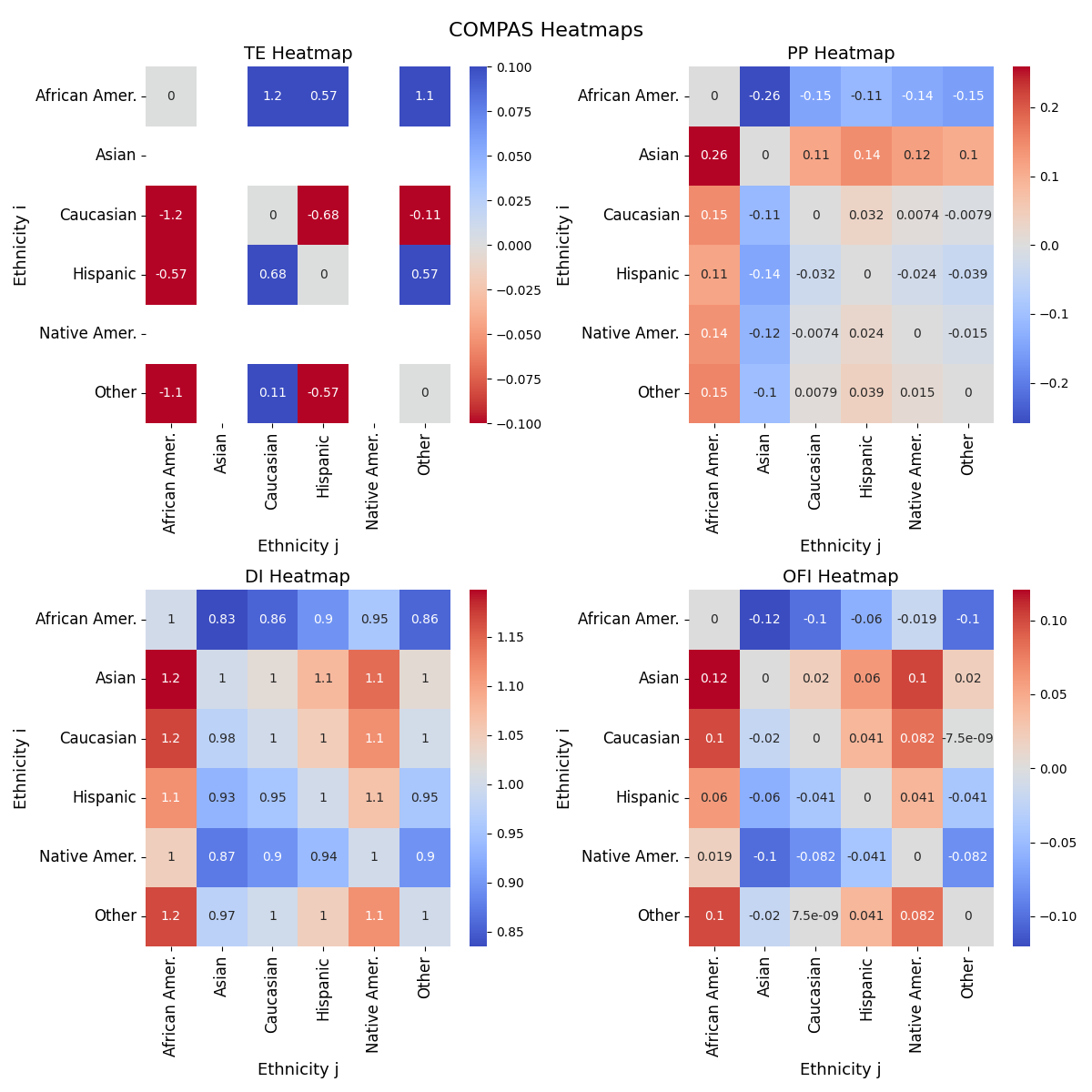}
    \caption{Assessing COMPAS with results published by Propublica\cite{propublica_compas}.}
    \label{fig:compas}
\end{figure*}

In Folktable’s cases, we investigate Georgia’s census database from 2014 to 2017 (a total of 393,236 observations) and predict if an individual is employed. Using the random forest classifier (RF), Figure \ref{fig:folkRf} reveals that OFI corroborates with DI in this case study, as it was in COMPAS. However, in the Na\"ive Bayes (NB) case study with Folktables, we find that OFI strongly disagrees with DI when context necessitates it.

\begin{figure*}
    \centering
\includegraphics[width=\linewidth]{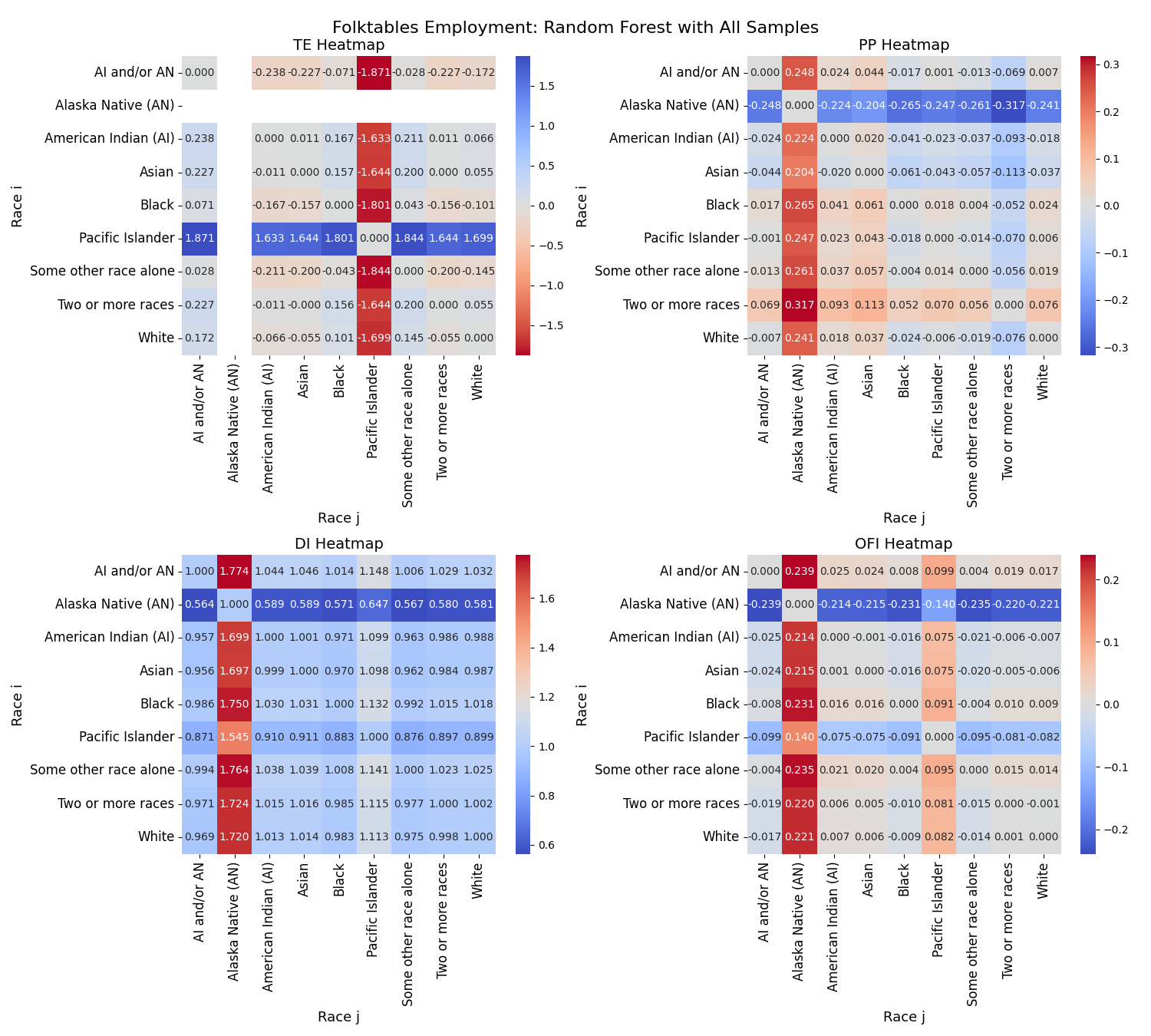}
    \caption{Assessing Folktables with a random forest classifier (all samples).}
    \label{fig:folkRf}
\end{figure*}

In the NB case study, OFI strongly disagrees with DI’s findings regarding Pacific Islanders. DI indicates that Pacific Islanders have a bias for them compared to 7/8 other races while OFI shows that, with context, Pacific Islanders have a bias for them over only 2/8 other races. The most prominent distinction occurs when comparing Pacific Islanders with Whites: DI suggests a large bias for Pacific Islanders while OFI shows a slight bias against them (Figure \ref{fig:folkNb}).

We further investigate OFI’s nuances by drawing 60 random samples to mimic a smaller dataset. In Figure \ref{fig:folkRfFew}, we see TE and DI’s problematic undefined cases. Furthermore, the “Two or more races” holds a constant 0 in DI, while OFI gives more information. OFI's nuance still shows substantial bias against the mixed-race class but indicates that they are more susceptible to bias when compared with whites as opposed to blacks.

Our practical use cases show that OFI gives unique insights into understanding the manifested bias. 

Notes: due to DI’s asymmetry, the heatmap's color map's gradient centers toward blue instead of gray like TE and OFI do. Furthermore, TE’s gradient is flipped as this definition has larger values indicating a bias toward Race $i$, unlike the others. Finally, since we assume the positive label is preferable, we flip the COMPAS labels so the positive is “not a recidivist”. 

\begin{figure*}
    \centering
\includegraphics[width=\linewidth]{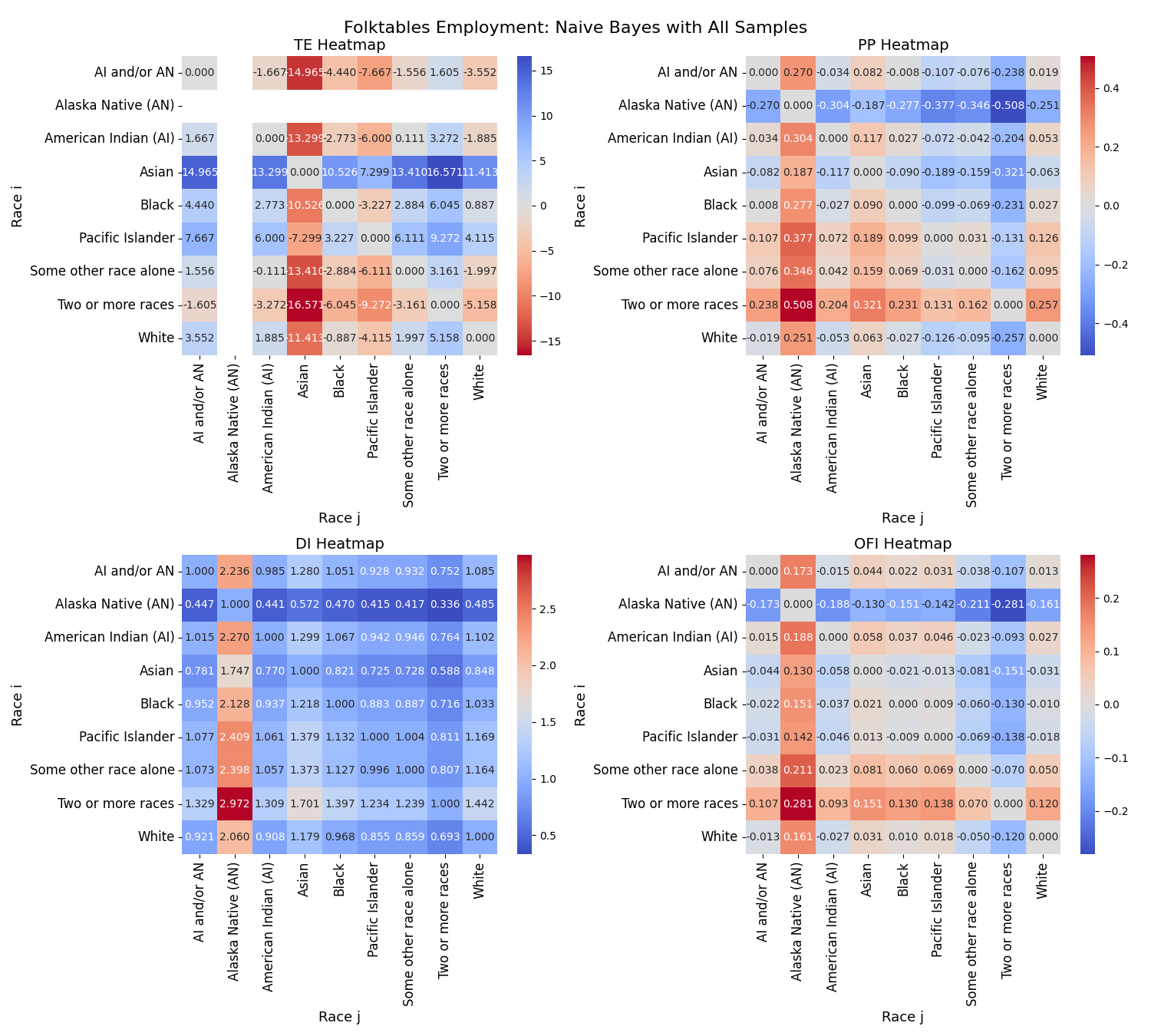}
    \caption{Assessing Folktables with Na\"ive Bayes.}
    \label{fig:folkNb}
\end{figure*}
\begin{figure*}
    \centering
    \includegraphics[width=\linewidth]{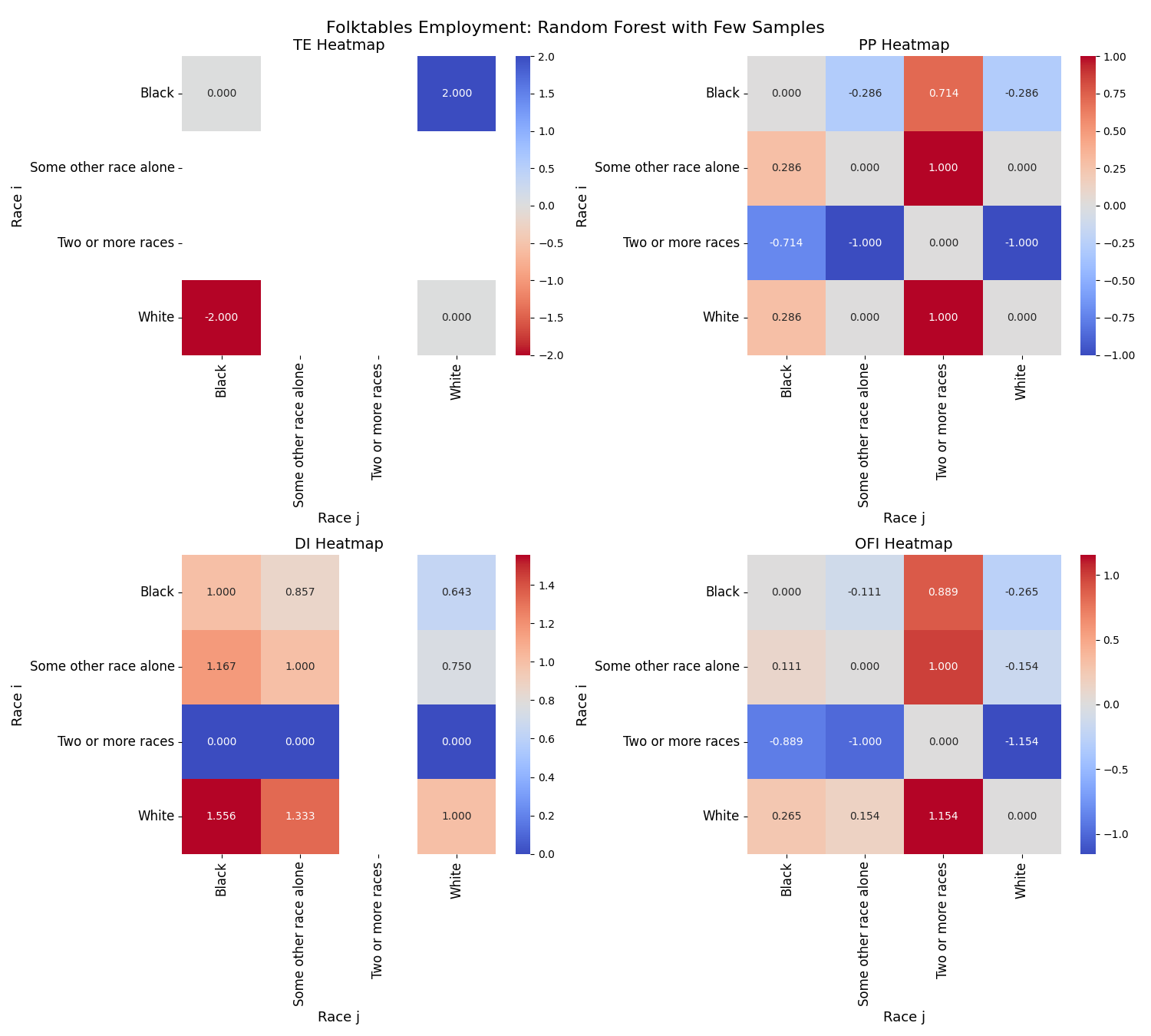}
    \caption{Assessing Folktables with a random forest classifier (60 samples).}
    \label{fig:folkRfFew}
\end{figure*}

\section*{Ethics Statement} This research is founded on the commitment to deepen the understanding of biases in machine learning through the critical lenses of legal philosophy and ethics. We emphasize that users of OFI should gain a thorough understanding of its meaning before taking corrective actions. Additionally, we recommend ongoing compliance checks with relevant local and international laws to ensure that the use of OFI adheres to evolving legal standards and contributes positively to the broader goal of fairness in machine learning applications. Through these practices, we aim to foster a responsible and legally consistent approach to the mitigation of bias in AI systems.

\end{document}